\documentclass[a4paper,reqno]{amsart}
\usepackage[all]{xy}           
\usepackage{amssymb}           
\usepackage{hyperref}
\usepackage{eucal}
\usepackage{graphicx}
\usepackage{epsfig}
\usepackage[usenames,dvipsnames]{color}
\numberwithin{equation}{section}

\newtheorem{definition}{Definition}[section]
\newtheorem{lemma}[definition]{Lemma}
\newtheorem{theorem}[definition]{Theorem}
\newtheorem{proposition}[definition]{Proposition}
\newtheorem{corollary}[definition]{Corollary}
\newtheorem{remarkth}[definition]{Remark}

\renewcommand{\emph}[1]{{\bfseries\itshape{#1}}}

\newcommand{\R}{\mathbb{R}}      

\newcount\ancho \newcount\anchom \newcount\anchoa
\newcount\anchob \newcount\altura

\newcommand{\ltilde}[3][0]{\altura=0 \advance\altura by #1
           \ancho=#2 \anchom=\ancho \divide\anchom by 2
           \anchoa=\ancho \divide\anchoa by 4
           \anchob=\anchom \advance\anchob by \anchoa
           \kern-3pt \begin{array}[b]{c}
           \begin{picture}(1,1)(\anchom,-\altura)
        \qbezier(0,2)(\anchoa,5)(\anchom,2)
        \qbezier(\anchom,2)(\anchob,-1)(\ancho,4)
        \qbezier(0,2)(\anchoa,4.5)(\anchom,1.8)
        \qbezier(\anchom,1.8)(\anchob,-1.5)(\ancho,4)
       \end{picture} \\[-4pt]{#3}
                       \end{array} \kern-4pt    }

\newcommand{\lhat}[3][0]{\altura=0 \advance\altura by #1
           \ancho=#2 \anchom=\ancho \divide\anchom by 2
           \anchoa=\ancho \divide\anchoa by 4
           \anchob=\anchom \advance\anchob by \anchoa
           \kern-3pt \begin{array}[b]{c}
           \begin{picture}(1,1)(\anchom,-\altura)
        \qbezier(0,2)(\anchoa,4)(\anchom,6)
        \qbezier(\anchom,6)(\anchob,4)(\ancho,2)
        \qbezier(0,2)(\anchoa,3.8)(\anchom,5.6)
        \qbezier(\anchom,5.6)(\anchob,3.8)(\ancho,2)
       \end{picture} \\[-4pt] {#3}
                       \end{array} \kern-4pt    }

\newcommand{\lcf}{\lbrack\! \lbrack}
\newcommand{\rcf}{\rbrack\! \rbrack}

\makeatletter
\newcommand\prol{\@ifstar{\@proldf}{\@prolpf}}  
\def\@prolpf{\@ifnextchar[{\@prolpf@wrt}{\@prolpf@}}
\def\@prolpf@wrt[#1]#2{\@ifnextchar[{\@prolpf@wrt@at{#1}{#2}}{\@prolpf@wrt@{#1}{#2}}}
\def\@prolpf@wrt@at#1#2[#3]{\prolsymbol^{#1}_{#3}#2}
\def\@prolpf@wrt@#1#2{\prolsymbol^{#1}#2}
\def\@prolpf@#1{\@ifnextchar[{\@prolpf@at{#1}}{\@prolpf@@{#1}}}
\def\@prolpf@at#1[#2]{\prolsymbol_{#2}#1}
\def\@prolpf@@#1{\prolsymbol#1}
\def\@proldf{\@ifnextchar[{\@proldf@wrt}{\@proldf@}}
\def\@proldf@wrt[#1]#2{\@ifnextchar[{\@proldf@wrt@at{#1}{#2}}{\@proldf@wrt@{#1}{#2}}}
\def\@proldf@wrt@at#1#2[#3]{\prolsymbol^{*#1}_{#3}#2}
\def\@proldf@wrt@#1#2{\prolsymbol^{*#1}#2}
\def\@proldf@#1{\@ifnextchar[{\@proldf@at{#1}}{\@proldf@@{#1}}}
\def\@proldf@at#1[#2]{\prolsymbol^*_{#2}#1}
\def\@proldf@@#1{\prolsymbol^*#1}
\def\prolsymbol{\mathcal{T}}
\makeatother

\def\lcf{\lbrack\! \lbrack}
\def\rcf{\rbrack\! \rbrack}
\begin{document}
{\Large

\title[Hamiltonian dynamics on Lie algebroids, unimodularity
and preservation of volumes]{Hamiltonian dynamics on Lie
algebroids, unimodularity and preservation of volumes}

\author[J.\ C.\ Marrero]{Juan C.\ Marrero}
\address{Juan C.\ Marrero:
ULL-CSIC Geometr\'{\i}a Diferencial y Mec\'anica Geom\'etrica\\
Departamento de Matem\'atica Fundamental, Facultad de
Ma\-te\-m\'a\-ti\-cas, Universidad de la Laguna, La Laguna,
Tenerife, Canary Islands, Spain} \email{jcmarrer@ull.es}

\thanks{This work has been partially supported by MEC (Spain)
Grants MTM 2006-03322 and project "Ingenio Mathematica" (i-MATH)
No. CSD 2006-00032 (Consolider-Ingenio 2010). The author would
like to thank Prof. Yuri Fedorov for stimulating discussions on
the subject. He also thanks the Department of Applied Mathematics
I, UPC (Barcelona), for the hospitality offered to him during a
visit to Barcelona where part of this work was done}



\begin{abstract}
In this paper we discuss the relation between the unimodularity of
a Lie algebroid $\tau_{A}: A \to Q$ and the existence of invariant
volume forms for the hamiltonian dynamics on the dual bundle
$A^*$. The results obtained in this direction are applied to
several hamiltonian systems on different examples of Lie
algebroids.
\end{abstract}

\maketitle

\tableofcontents

\section{Introduction}
It is well-known that Hamilton equations in Classical Mechanics
may be written in an intrinsic way using the canonical symplectic
structure $\omega_{Q}$ of the phase space of momenta $T^*Q$. In
fact, if $H$ is the hamiltonian energy, the sum of the kinetic and
the potential energy, then the solutions of the Hamilton equations
for $H$ are the integral curves of the hamiltonian vector field
${\mathcal H}_{H}^{\Pi_{T^*Q}}$ of $H$ with respect to
$\omega_{Q}$. The flow of ${\mathcal H}_{H}^{\Pi_{T^*Q}}$
preserves the symplectic form. Thus, one directly deduces {\em
Liouville's theorem}: {\em the flow of the hamiltonian vector
field preserves the symplectic volume} (see, for instance,
\cite{AbMa}).

However, for a hamiltonian system on a general Poisson manifold,
not necessarily symplectic, the flow of the hamiltonian vector
field doesn't preserves, in general, a volume form on the phase
space. We remark that the existence of invariant volume forms is
interesting for reasons of integrability (see \cite{Ko1} and the
references therein). So, it is important to obtain necessary and
sufficient conditions for a hamiltonian system on a Poisson
manifold admits an invariant volume. In this direction, a very
nice result was proved by Kozlov \cite{Ko}: {\em for a hamiltonian
system of kinetic type on the dual space ${\frak g}^*$ of a Lie
algebra ${\frak g}$ the flow of the corresponding hamiltonian
vector field preserves a volume form on ${\frak g}^*$ if and only
if the Lie algebra ${\frak g}$ is unimodular}. Note that, in this
case, the Poisson structure on ${\frak g}^*$ is the Lie-Poisson
structure induced by the Lie algebra structure of ${\frak g}$.

Hamiltonian systems on the dual space of a Lie algebra ${\frak g}$
arise, in a natural way, as the reduction of standard
left-invariant hamiltonian systems with configuration space a Lie
group. More generally, hamiltonian systems on Poisson manifolds
arise, in a natural way, as the reduction of standard hamiltonian
systems which are invariant under the action of a symmetry Lie
group $G$. The typical situation is the following one (see, for
instance, \cite{OrRa}). The configuration space $Q$ of the
mechanical system is the total space of a principal $G$-bundle
$p:Q \to Q/G$. So, the reduced phase space is a vector bundle
$T^*Q/G$ over $Q/G$ and the Poisson structure on it is linear.
This implies that the dual bundle $TQ/G$ (over $Q/G$) admits a Lie
algebroid structure. In fact, the vector bundle $\tau_{TQ/G}: TQ/G
\to Q/G$ is just the Atiyah algebroid associated with the
principal $G$-bundle $p: Q\to Q/G$ (see \cite{Mac}).

We recall that Lie algebroids are a natural generalization of
tangent bundles and Lie algebras and that there exists a
one-to-one correspondence between Lie algebroid structures on a
vector bundle $\tau_{A}: A \to Q$ and linear Poisson structures on
the dual bundle $A^*$ (see \cite{Co,Mac}). Thus, if we have a Lie
algebroid structure on the vector bundle $\tau_{A}: A \to Q$ and a
hamiltonian function $H$ of mechanical type on $A^*$, we can
consider the hamiltonian vector field of $H$ with respect to the
linear Poisson structure $\Pi_{A^*}$ on $A^*$ and the
corresponding dynamical system on $A^*$. Using this procedure, one
may recover the standard Hamilton equations (when $A$ is the
standard Lie algebroid $\tau_{TQ}: TQ \to Q$), the Lie-Poisson
equations (when $A$ is a Lie algebra as a Lie algebroid over a
single point), the Hamilton-Poincar\'e equations (when $A$ is the
Atiyah algebroid associated with a principal $G$-bundle), the Lie
Poisson equations on the dual of a semidirect product of Lie
algebras (when $A$ is an action Lie algebroid over a real vector
space),... (see \cite{LeMaMa,We0}).

So, a natural problem arise: to find necessary and sufficient
conditions for a hamiltonian system on the dual bundle to a Lie
algebroid $\tau_{A}: A \to Q$ admits an invariant volume.

In this paper, we will obtain such conditions. For this purpose,
we will use a geometrical object associated with the Lie
algebroid: {\em the modular class} of $\tau_{A}: A \to Q$.

The modular class of a Lie algebroid $A$ was introduced in
\cite{EvLuWe} (see also \cite{We}) as follows. We will assume that
$Q$ and the vector bundle are orientable and we fix a volume form
$\nu$ on $Q$ and a section $\Lambda \in \Gamma(\Lambda^nA)$ such
that $\Lambda_{q} \neq 0$, for all $q\in Q$, where $n$ is the rank
of $A$. Then, $\nu$ and $\Lambda$ induce a volume form on $\nu
\wedge \Lambda$ on $A^*$ and the modular section of $A$ with
respect to $\nu$ and $\Lambda$ is a section of $\tau_{A^*}: A^*
\to Q$ whose vertical lift is just the modular vector field of
$\Pi_{A^*}$ with respect to the volume form $\nu\wedge \Lambda$.
We remark that such a vector field is defined using the divergence
of the hamiltonian vector fields on $A^*$ with respect to the
volume $\nu\wedge \Lambda$. The modular section defines a
cohomology class in the cohomology complex of $A$ with trivial
coefficients. This cohomology class doesn't depend on the volumes
$\nu$ and $\Lambda$ and it is called the modular class of the Lie
algebroid $A$. $A$ is said to be unimodular if its modular class
vanishes. The standard Lie algebroid $\tau_{TQ}: TQ \to Q$ is
unimodular and a Lie algebra ${\frak g}$ is unimodular as a Lie
algebroid over a single point if and only if it is unimodular in
the classical sense, that is, the modular character of ${\frak g}$
is zero (for more details, see \cite{EvLuWe} and Section
\ref{sec3.2}).

Using the modular class of the Lie algebroid $\tau_{A}: A \to Q$
we deduce the main results of the paper. In fact, if $H: A^* \to
\R$ is a hamiltonian function of mechanical type, we prove the
following facts:
\begin{itemize}
\item
{\em $A$ is unimodular if and only if the hamiltonian dynamics
preserves a volume form on $A^*$ of basic type.}

We remark that a volume form $\Phi = e^{\tilde{\sigma}}\nu \wedge
\Lambda$ on $A^*$, with $\tilde{\sigma} \in C^{\infty}(A^*)$, is
of basic type if $\tilde{\sigma}$ is a basic function, that is,
$\tilde{\sigma} \in C^{\infty}(Q)$. The previous result
generalizes Liouville's theorem (note that the standard Lie
algebroid $\tau_{TQ}: TQ \to Q$ is unimodular).

\item
{\em If the potential energy is constant then $A$ is unimodular if
and only if the hamiltonian dynamics preserves a volume form on
$A^*$.}

It is clear that if we apply this result to the particular case
when $A$ is a Lie algebra we recover Kozlov's theorem.
\end{itemize}

The paper is organized as follows. In Sections 2 and 3, we recall
some definitions and results on Lie algebroids, hamiltonian
dynamics and modular class which will be used in the rest of the
paper. Section 4 contains the main results. In fact, in this
section we discuss the relation between the unimodularity of a Lie
algebroid $A$ and the existence of invariant volumes for the
hamiltonian dynamics on $A^*$ (see Theorem \ref{maintheorem} and
Corollaries \ref{Corolario1} and \ref{Corolario2}). In Section 5,
we apply the results of Section 4 to several hamiltonian systems
on different examples of Lie algebroids. The paper ends with our
conclusions and a description of future research directions.

\section{Lie algebroids}\label{section2}

\subsection{Definitions and notation}
Let $\tau_{A}: A \to Q$ be a vector bundle of rank $n$ over a
manifold $Q$ of dimension $m$. Denote by $\Gamma(\tau_{A})$ the
space of sections of the vector bundle $\tau_{A}: A \to Q$.

\begin{definition}
A Lie algebroid structure on the vector bundle $\tau_{A}: A \to Q$
is a Lie bracket $\lcf \cdot , \cdot \rcf_{A}: \Gamma(\tau_{A})
\times \Gamma(\tau_{A}) \to \Gamma(\tau_{A})$ on the space
$\Gamma(\tau_{A})$ and a vector bundle morphism $\rho_{A}:A \to
TQ$, the anchor map, such that if we also denote by $\rho_{A}:
\Gamma(\tau_{A}) \to {\frak X}(Q)$ the corresponding morphism of
$C^{\infty}(Q)$-modules then
\[
\lcf X, fY \rcf_{A} = f \lcf X, Y\rcf_{A} + \rho_{A}(X)(f)Y,
\mbox{ for } X, Y \in \Gamma(\tau_{A}) \mbox{ and } f \in
C^{\infty}(Q).
\]

\end{definition}
If $(\lcf \cdot , \cdot \rcf_{A}, \rho_{A})$ is a Lie algebroid
structure on the vector bundle $\tau_{A}: A \to Q$ it follows that
\begin{equation}\label{Homomor}
\rho_{A}\lcf X, Y \rcf_{A} = [\rho_{A}(X), \rho_{A}(Y)], \mbox{
for } X, Y \in \Gamma(\tau_{A})
\end{equation}
Moreover, if $(q^i)$ are local coordinates in an open subset $U$
of $Q$ and $\{e_{\alpha}\}$ is a basis of sections of the vector
bundle $\tau_{A}^{-1}(U) \to U$, we have that
\[
\lcf e_{\alpha}, e_{\beta}\rcf_{A} =
C_{\alpha\beta}^{\gamma}e_{\gamma}, \makebox[.6cm]{}
\rho_{A}e_{\alpha} = \rho^i_{\alpha}\displaystyle
\frac{\partial}{\partial q^i},
\]
with $C_{\alpha\beta}^{\gamma}, \rho^{i}_{\alpha} \in
C^{\infty}(U)$. The functions $C_{\alpha\beta}^{\gamma},
\rho^{i}_{\alpha}$ are called {\em the local structure functions}
of the Lie algebroid with respect to the local coordinates $(q^i)$
and the basis $\{e_{\alpha}\}$. Using (\ref{Homomor}) and the fact
that $\lcf \cdot , \cdot \rcf_{A}$ is a Lie bracket, we deduce
that
\begin{equation}\label{estruc1}
\rho_\alpha^j\frac{\partial \rho_\beta^i}{\partial q^j}
-\rho_\beta^j\frac{\partial \rho_\alpha^i}{\partial q^j}=
\rho_\gamma^iC_{\alpha\beta}^\gamma \end{equation}

and

\begin{equation}\label{estruc2}
\sum_{cyclic(\alpha,\beta,\gamma)}[\rho_{\alpha}^i\frac{\partial
C_{\beta\gamma}^\nu}{\partial q^i} + C_{\alpha\mu}^\nu
C_{\beta\gamma}^\mu]=0.
\end{equation}
These equations are called {\em the local structure equations} of
the Lie algebroid with respect to the local coordinates $(q^i)$
and the basis $\{e_{\alpha}\}$.

\subsection{Examples}\label{Sec2.2}
Next, we will exhibit some examples of Lie algebroids.

\medskip

\noindent {\bf The Atiyah algebroid associated with a principal
$G$-bundle}. Let $p:Q \to Q/G$ be a principal $G$-bundle. Then, we
may consider the tangent lift of the principal action of $G$ on
$Q$ and, it is well-known that, the space of orbits of this
action, $TQ/G$, is a real vector bundle over $Q/G$. The vector
bundle projection $\tau_{TQ/G}: TQ/G \to Q/G$ is given by
\[
\tau_{TQ/G}[v_{q}] = p(q), \makebox[.6cm]{} \forall v_q \in T_qQ.
\]
Furthermore, the space of sections $\Gamma(\tau_{TQ/G})$ may be
identified with the set of $G$-invariant vector fields on $Q$.
Thus, using that the Lie bracket of two $G$-invariant vector
fields on $Q$ also is $G$-invariant, we may define, in a natural
way, a Lie bracket on the space $\Gamma(\tau_{TQ/G})$
\[
\lcf \cdot , \cdot \rcf_{TQ/G}: \Gamma(\tau_{TQ/G}) \times
\Gamma(\tau_{TQ/G}) \to \Gamma(\tau_{TQ/G}).
\]
On the other hand, the anchor map $\rho_{TQ/G}: TQ/G \to T(Q/G)$
is given by
\[
\rho_{TQ/G}[v_{q}] = (T_qp)(v_q), \mbox{ for } v_{q} \in T_qQ,
\]
where $Tp: TQ \to T(Q/G)$ is the tangent map to the principal
bundle projection $p: Q \to Q/G$.

The resultant Lie algebroid $(TQ/G, \lcf \cdot , \rcf _{TQ/G},
\rho_{TQ/G})$ is called {\em the Atiyah algebroid} associated with
the principal $G$-bundle $p:Q \to Q/G$ (see \cite{Mac}).

Note that if the Lie group $G$ is trivial and $p = id$ then the
Atiyah algebroid may be identified with the standard Lie algebroid
$\tau_{TQ}: TQ \to Q$.

Another interesting particular case is when the manifold $Q$ is
the Lie group $G$. In this case, using that $TG$ is diffeomorphic
to the product manifold $G \times {\frak g}$ (${\frak g}$ being
the Lie algebra of $G$), we deduce that the Atiyah algebroid
$TG/G$ may be identified with ${\frak g}$ (as a Lie algebroid over
a single point).

Finally, if the principal $G$-bundle is trivial, that is,
\[
Q = G \times M, \makebox[.75cm]{} Q/G = M,
\]
$p$ is the canonical projection on the second factor and the
action of $G$ on $Q = G \times M$ is defined by
\[
g(g', x) = (gg', x), \mbox{ for } g, g' \in G \mbox{ and } x\in M,
\]
then it is easy to prove that the quotient vector bundle
$\tau_{TQ/G}: TQ/G \to Q/G$ may be identified with the vector
bundle $\tau_{{\frak g} \times TM}: {\frak g} \times TM \to M$.
Under this identification, the Lie algebroid structure $(\lcf
\cdot , \cdot \rcf_{{\frak g} \times TM}, \rho_{{\frak g} \times
TM})$ on $\tau_{{\frak g}\times TM}: {\frak g}\times TM \to M$ is
given by
\[
\lcf (\xi, X), (\eta, Y)\rcf_{{\frak g} \times TM} = ([\xi,
\eta]_{{\frak g}}, [X, Y]), \makebox[.5cm]{} \rho_{{\frak g}
\times TM}(\xi, X) = X,
\]
for $\xi, \eta \in {\frak g}$ and $X, Y \in {\frak X}(M)$,
$[\cdot, \cdot ]_{{\frak g}}$ being the Lie bracket on ${\frak
g}$.

\medskip

\noindent{\bf The Lie algebroid associated with a left
infinitesimal action}. Let ${\frak g}$ be a real Lie algebra of
finite dimension and $\Phi: {\frak g} \to {\frak X}(Q)$ be a left
infinitesimal action of ${\frak g}$ on $Q$, that is, $\Phi$ is a
$\R$-linear map and
\[
\Phi([\xi, \eta]_{{\frak g}}) = -[\Phi(\xi), \Phi(\eta)], \mbox{
for } \xi, \eta \in {\frak g}.
\]
Then, the trivial vector bundle $\tau_{A}: A = {\frak g} \times Q
\to Q$ is a Lie algebroid. In fact, the anchor map $\rho_{A}: A =
{\frak g} \times Q \to TQ$ is given by
\[
\rho_{A}(\xi, q) = -\Phi(\xi)(q), \mbox{ for } (\xi, q) \in A=
{\frak g} \times Q.
\]
On the other hand, it is clear that the space of sections
$\Gamma(\tau_{A})$ may be identified with the set $C^{\infty}(Q,
{\frak g})$ of smooth functions from $Q$ on ${\frak g}$. Under
this identification, the Lie bracket $\lcf \cdot , \cdot \rcf_{A}:
\Gamma(\tau_{A}) \times \Gamma(\tau_{A}) \to \Gamma(\tau_{A})$ is
defined by
\[
\lcf \varphi, \psi \rcf_{A}(q) = [\varphi(q), \psi(q)]_{{\frak g}}
- \Phi(\varphi(q))(q)(\psi) + \Phi(\psi(q))(q)(\varphi),
\]
for $\varphi, \psi \in C^{\infty}(Q, {\frak g})$ and $q \in Q$.

The resultant Lie algebroid is called {\em the Lie algebroid
associated with the left infinitesimal action $\Phi$} (see
\cite{HiMa}).

\subsection{The differential associated with a Lie algebroid}
Let $(\lcf \cdot , \cdot \rcf_{A}, \rho_{A})$ be a Lie algebroid
structure on a vector bundle $\tau_{A}: A \to Q$ of rank $n$ and
$\tau_{A^*}: A^{*} \to Q$ be the dual vector bundle to $\tau_{A}:
A \to Q$. Then, one may introduce the corresponding {\em
differential} as a $\R$-linear map $d^A:
\Gamma(\Lambda^k\tau_{A^*}) \to \Gamma(\Lambda^{k+1}\tau_{A^*})$,
$k \in \{0, \dots , n-1\}$, given by
\begin{equation}\label{deD}
\begin{array}{rcl}
(d^{A}\alpha)(X_{0}, X_{1}, \dots, X_{k}) &=& \displaystyle
\sum_{i=0}^{k} (-1)^{i} \rho_{A}(X_{i})(\alpha(X_{0}, \dots,
\hat{X}_{i}, \dots, X_{k})) \\[5pt]
&& + \displaystyle \sum_{i < j} (-1)^{i+j} \alpha (\lcf X_{i},
X_{j} \rcf_{A}, X_{0}, X_{1}, \dots, \hat{X}_{i}, \dots,
\hat{X}_{j}, \dots, X_{k})
\end{array}
\end{equation}
for $\alpha \in \Gamma(\Lambda^k \tau_{A^*})$ and $X_{0}, X_{1},
\dots , X_{k} \in \Gamma(\tau_{A})$.

We have that
\[
d^A(\alpha \wedge \beta) = d^A\alpha \wedge \beta + (-1)^k \alpha
\wedge d^A\beta, \makebox[.75cm]{} (d^A)^2 = 0,
\]
for $\alpha \in \Gamma(\Lambda^k \tau_{A^*})$ and $\beta \in
\Gamma(\Lambda^r \tau_{A^*})$. Thus, we can consider the
corresponding cohomology groups $H^kA = \displaystyle\frac{Ker \;
d^A}{Im \; d^A}$, for $k \in \{0, 1, \dots , n\}$ (see
\cite{Mac}).

Moreover, if $(q^i)$ are local coordinates on an open subset $U$
of $Q$ and $\{e_{\alpha}\}$ is a basis of sections of the vector
bundle $\tau_{A}^{-1}(U) \to U$, it follows that
\[
d^Aq^i= \rho^i_{\alpha}e^{\alpha}, \makebox[.75cm]{} d^Ae^{\gamma}
= \displaystyle -\frac{1}{2}
C_{\alpha\beta}^{\gamma}e^{\alpha}\wedge e^{\beta},
\]
where $\{e^{\alpha}\}$ is the dual basis of $\{e_{\alpha}\}$ and
$\rho^i_{\alpha}$, $C_{\alpha\beta}^{\gamma}$ are the local
structure functions of $A$.

On the other hand, if $X \in \Gamma(\tau_{A})$ one may define {\em
the Lie derivative operator}
\[
{\mathcal L}_{X}^{A}: \Gamma(\Lambda^k\tau_{A^*}) \to
\Gamma(\Lambda^k\tau_{A^*})
\]
as follows
\[
{\mathcal L}_{X}^A = i_{X} \circ d^A + d^A \circ i_{X},
\]
where $i_{X}: \Gamma(\Lambda^k\tau_{A^*}) \to
\Gamma(\Lambda^{k-1}\tau_{A^*})$ is the contraction by $X$, that
is,
\[
(i_{X}\alpha)(X_{1}, \dots , X_{k-1}) = \alpha(X, X_{1}, \dots ,
X_{k-1}),
\]
for $X_{1}, \dots , X_{k-1} \in \Gamma(\tau_{A})$.

\subsection{The linear Poisson structure associated with a Lie
algebroid} Let $(\lcf \cdot , \cdot \rcf_{A}, \rho_{A})$ be a Lie
algebroid structure on the vector bundle $\tau_{A}: A \to Q$.
Then, on may define a $\R$-linear bracket of functions
\[
\{\cdot , \cdot \}_{A^*}: C^{\infty}(A^*) \times C^{\infty}(A^*)
\to C^{\infty}(A^*)
\]
which is characterized by the following conditions
\[
\{\hat{X}, \hat{Y}\}_{A^*} = -\widehat{\lcf X, Y\rcf}_{A}, \; \;
\; \{f \circ \tau_{A^*}, \hat{X}\}_{A^*} = \rho_{A}(X)(f) \circ
\tau_{A^*},
\]
and
\[
\{f \circ \tau_{A^*}, g \circ \tau_{A^*}\}_{A^*} = 0,
\]
for $X, Y \in \Gamma(\tau_{A})$ and $f, g \in C^{\infty}(Q)$. Note
that if $X \in \Gamma(\tau_{A})$ then $\hat{X}: A^* \to \R$ is the
linear function on $A^*$ given by
\[
\hat{X}(\alpha) = \alpha(X(\tau_{A^*}(\alpha))), \mbox{ for }
\alpha \in A^*.
\]
$\{\cdot , \cdot \}_{A^*}$ is a {\em linear Poisson structure} on
$A^*$, that is,
\begin{enumerate}
\item
$\{\cdot, \cdot \}_{A^*}$ is a Lie bracket on $C^{\infty}(A^*)$,

\item
$\{\cdot, \cdot\}_{A^*}$ satisfies the Leibniz rule
\[
\{\varphi\varphi', \psi\}_{A^*} = \varphi\{\varphi', \psi\}_{A^*}
+ \varphi'\{\varphi, \psi\}_{A^*},
\]
for $\varphi, \varphi', \psi \in C^{\infty}(A^*)$ and

\item
$\{\cdot, \cdot \}_{A^*}$ is linear or, in other words, the
bracket of two linear functions on $A^*$ is a linear function.

\end{enumerate}

We will denote by $\Pi_{A^*}$ the corresponding linear Poisson
$2$-vector on $A^*$ which is characterized by the following
condition
\[
\Pi_{A^*}(d\varphi, d\psi) = \{\varphi, \psi\}_{A^*}
\]
(for more details, see \cite{Co}).

If $(q^i)$ are local coordinates on an open subset $U$ of $Q$ and
$\{e_{\alpha}\}$ is a basis of sections of the vector bundle
$\tau_{A}^{-1}(U) \to U$ we have that
\begin{equation}\label{PiAstar}
\Pi_{A^*} = \displaystyle \rho^i_{\alpha} \frac{\partial}{\partial
q^i} \wedge \frac{\partial}{\partial p_{\alpha}} - \frac{1}{2}
C_{\alpha\beta}^{\gamma}p_{\gamma} \frac{\partial}{\partial
p_{\alpha}} \wedge \frac{\partial}{\partial p_{\beta}}
\end{equation}
where $(q^i, p_{\alpha})$ are the corresponding local coordinates
on $A^*$ and $\rho^i_{\alpha}, C_{\alpha\beta}^{\gamma}$ are the
local structure functions of $A$.

Next, we will describe the linear Poisson structure on the dual
vector bundle of some examples of Lie algebroids.

\begin{itemize}
\item
If $A$ is the standard Lie algebroid $\tau_{TQ}: TQ \to Q$ then
the local structure functions of $\tau_{TQ}: TQ \to Q$ with
respect to the local coordinates $(q^i)$ on $Q$ and the local
basis $\{\frac{\partial}{\partial q^i}\}$ are
\[
\rho^i_{\alpha} = \delta^{i}_{\alpha}, \makebox[.75cm]{}
C_{\alpha\beta}^{\gamma} = 0.
\]
Thus, $\{\cdot, \cdot\}_{T^*Q}$ is the linear Poisson bracket
induced by {\em the canonical symplectic structure} of $T^*Q$ (for
the definition of this structure see, for instance, \cite{AbMa}).

\item
If $A$ is a real Lie algebra ${\frak g}$ of finite dimension then
$\rho_{A} = \rho_{{\frak g}}$ is the zero map and $\{\cdot,
\cdot\}_{\frak g^*}$ is {\em the Lie-Poisson structure} on ${\frak
g}^*$ induced by the Lie algebra ${\frak g}$ (for the definition
of this structure see, for instance, \cite{MaRa}).

\item
If $p: G\times M \to M$ is a trivial principal $G$-bundle then the
dual vector bundle of the Atiyah algebroid $\tau_{{\frak g} \times
TM}: {\frak g} \times TM \to M$ is the vector bundle $\tau_{{\frak
g}^* \times T^*M}: {\frak g}^* \times T^*M \to M$ and the linear
Poisson structure on ${\frak g}^* \times T^*M$ is just the product
of the Lie-Poisson structure on ${\frak g}^*$ and the canonical
symplectic structure of $T^*M$.

\item
If $\Phi: {\frak g} \to {\frak X}(Q)$ is a left inifinitesimal
action of a Lie algebra ${\frak g}$ on a manifold $Q$ then the
dual vector bundle of the corresponding action Lie algebroid is
the vector bundle $\tau_{Q \times {\frak g}^*}: Q \times {\frak
g}^* \to Q$ and the linear Poisson structure on $Q \times {\frak
g}^*$ is characterized by the following conditions
\[
\{\hat{\xi}, \hat{\eta}\}_{Q\times {\frak g}^*} = -\widehat{[\xi,
\eta]}_{\frak g}, \makebox[.5cm]{} \{f \circ \tau_{Q \times {\frak
g}^*}, \hat{\eta}\}_{Q\times {\frak g}^*} = \Phi(\eta)(f) \circ
\tau_{Q \times {\frak g}^*},
\]
and
\[
\{f \circ \tau_{Q \times {\frak g}^*}, g \circ \tau_{Q \times
{\frak g}^*}\}_{Q \times {\frak g}^*} = 0,
\]
for $f, g \in C^{\infty}(Q)$ and $\xi, \eta \in {\frak g}$.

\end{itemize}

\section{Hamiltonian dynamics on Lie algebroids and modular
sections}

\subsection{Hamiltonian dynamics on Lie algebroids}
Let $(\lcf \cdot, \cdot \rcf_{A}, \rho_{A})$ be a Lie algebroid
structure on a vector bundle $\tau_{A}: A \to Q$ and $\Pi_{A^*}$
be the corresponding linear Poisson $2$-vector on $A^*$.

If $H: A^* \to \R$ is a hamiltonian function on $A^*$ then one may
consider the hamiltonian vector field ${\mathcal
H}_{H}^{\Pi_{A^*}}$ of $H$ with respect to the Poisson structure
$\Pi_{A^*}$, that is,
\begin{equation}\label{Hamiltonian} {\mathcal
H}_{H}^{\Pi_{A^*}}(F) = \{F, H\}_{A^*} = \Pi_{A^*}(dF, dH), \mbox{
for } F \in C^{\infty}(A^*).
\end{equation}
The solutions of {\em the Hamilton equations} for $H$ are just the
integral curves of ${\mathcal H}_{H}^{\Pi_{A^*}}$.

From (\ref{PiAstar}) and (\ref{Hamiltonian}), we deduce that the
local expression of ${\mathcal H}_{H}^{\Pi_{A^*}}$ is
\begin{equation}\label{HamH}
{\mathcal H}_{H}^{\Pi_{A^*}} = \displaystyle \frac{\partial
H}{\partial p_{\alpha}}\rho^i_{\alpha} \frac{\partial}{\partial
q^i} - (\frac{\partial H}{\partial q^i}\rho^i_{\alpha} +
\frac{\partial H}{\partial
p_{\beta}}C_{\alpha\beta}^{\gamma}p_{\gamma})\frac{\partial}{\partial
p_{\alpha}}.
\end{equation}
Thus, the Hamilton equations are
\[
\displaystyle \frac{dq^i}{dt} = \frac{\partial H}{\partial
p_{\alpha}}\rho^i_{\alpha}, \makebox[.5cm]{}
\frac{dp_{\alpha}}{dt} = - (\frac{\partial H}{\partial
q^i}\rho^i_{\alpha} + \frac{\partial H}{\partial
p_{\beta}}C_{\alpha\beta}^{\gamma}p_{\gamma})
\]
(for more details, see \cite{LeMaMa}).

If $A$ is the standard Lie algebroid $\tau_{TQ}: TQ \to Q$ the
Hamilton equations for $H: T^*Q \to \R$ are
\[
\displaystyle \frac{dq^i}{dt} = \frac{\partial H}{\partial p_{i}},
\makebox[.5cm]{} \frac{dp_{i}}{dt} = - \frac{\partial H}{\partial
q^i},
\]
that is, {\em the standard Hamilton equations} (see, for instance,
\cite{AbMa}).

For a real Lie algebra ${\frak g}$ of finite dimension we have
that the Hamilton equations for $H: {\frak g}^* \to \R$ are
\begin{equation}\label{Lie-Poisson}
\displaystyle \frac{dp_{\alpha}}{dt} = - \frac{\partial
H}{\partial p_{\beta}}c_{\alpha\beta}^{\gamma}p_{\gamma}
\end{equation}
where $c_{\alpha\beta}^{\gamma}$ are the structure constants of
${\frak g}$ with respect to a basis $\{e_{\alpha}\}$. Note that
(\ref{Lie-Poisson}) are just {\em the Lie-Poisson equations} for
$H$ (see, for instance, \cite{MaRa}).

If $H: {\frak g}^* \times T^*M \to \R$ is a hamiltonian function
on the dual bundle of the Atiyah algebroid $\tau_{{\frak g} \times
TM}: {\frak g} \times TM \to M$ then the Hamilton equations for
$H$ are
\begin{equation}\label{Ham-Poincare}
\displaystyle \frac{dq^i}{dt} = \frac{\partial H}{\partial p_{i}},
\makebox[.4cm]{} \frac{dp_{i}}{dt} = - \frac{\partial H}{\partial
q^i}, \makebox[.4cm]{} \frac{dp_{\alpha}}{dt} = - \frac{\partial
H}{\partial p_{\beta}}c_{\alpha\beta}^{\gamma}p_{\gamma},
\end{equation}
where $(p_{\alpha}, q^i, p_{i})$ are local coordinates on ${\frak
g}^* \times T^*M$ and $c_{\alpha\beta}^{\gamma}$ are the structure
constants of ${\frak g}$ with respect to a basis of ${\frak g}$.
(\ref{Ham-Poincare}) are just {\em the Hamilton-Poincar\'e
equations} for $H$ (see, for instance, \cite{LeMaMa}).

\subsection{Modular class of a Lie
algebroid}\label{modularclass}\label{sec3.2}
 Let $(\lcf\cdot, \cdot \rcf_{A}, \rho_{A})$ be a Lie algebroid
 structure on a vector bundle $\tau_{A}: A \to Q$ of rank $n$ with
 base manifold $Q$ of dimension $m$.

We will assume that $Q$ and the vector bundle $\tau_{A}: A \to Q$
are orientable.

If $\alpha$ is a section of $\tau_{A^*}: A^* \to Q$ we will denote
by $\alpha^{\bf v} \in {\frak X}(A^*)$ {\em the vertical lift} of
$\alpha$. We recall that
\[
\alpha^{\bf v}(\gamma_{q}) = \displaystyle
\frac{d}{dt}_{|t=0}(\gamma_{q} + t\alpha(q)), \mbox{ for }
\gamma_{q} \in A^*_q.
\]

\begin{lemma}\label{Voldual}
Let $\nu$ be a volume form on $Q$ and $\Lambda$ be a section of
the vector bundle $\tau_{\Lambda^nA}: \Lambda^nA \to Q$ such that
$\Lambda(q) \neq 0$, for all $q\in Q$. Then, there exists a unique
volume form $\nu \wedge \Lambda$ on the dual bundle $A^*$ to the
vector bundle $\tau_{A}: A \to Q$ such that
\begin{equation}\label{Defvol}
\nu \wedge \Lambda(\tilde{Z}_{1}, \dots , \tilde{Z}_{m},
\alpha_{1}^{\bf v}, \dots , \alpha_{n}^{\bf v}) = \nu(Z_{1}, \dots
, Z_{m}) \Lambda(\alpha_{1}, \dots , \alpha_{n}),
\end{equation}
for $\alpha_{1}, \dots , \alpha_{n} \in \Gamma(\tau_{A^*})$ and
$\tilde{Z}_{1}, \dots , \tilde{Z}_{m}$ vector fields on $A^*$
which are $\tau_{A^*}$-projectable on the vector fields $Z_{1},
\dots , Z_{m}$ on $Q$.
\end{lemma}
\begin{proof}
We can choose local coordinates $(q^i)$ on an open subset $U$ of
$Q$ and a basis of sections $\{e_{\alpha}\}$ of the vector bundle
$\tau_{A}^{-1}(U) \to U$ such that on $U$
\[
\nu = e^{\sigma_{U}^{\nu}} dq^1 \wedge \dots \wedge dq^m,
\makebox[.75cm]{} \Lambda = e_{1} \wedge \dots \wedge e_{n}, \; \;
\mbox{ with } \sigma_{U}^{\nu} \in C^{\infty}(U).
\]
If $(q^i, p_{\alpha})$ are the corresponding local coordinates on
$A^*$ then on the open subset $\tau_{A^*}^{-1}(U)$ of $A^*$ we
consider the volume form
\[
(\nu \wedge \Lambda)_{\tau_{A^*}^{-1}(U)} = e^{\sigma_{U}^{\nu}}
dq^1 \wedge \dots \wedge dq^m \wedge dp_{1} \wedge \dots \wedge
dp_{n}.
\]
A direct computation proves that $(\nu \wedge
\Lambda)_{\tau_{A^*}^{-1}(U)}$ satisfies (\ref{Defvol}).

On the other hand, if $\Phi$ is a volume form on
$\tau_{A^*}^{-1}(U)$ such that (\ref{Defvol}) holds then it
follows that
\[
\Phi = (\nu \wedge \Lambda)_{\tau_{A^*}^{-1}(U)}.
\]
This proves the result.
\end{proof}

Suppose that $\nu$ is a volume form on $Q$ and that $\Lambda \in
\Gamma(\tau_{\Lambda^{n}A})$ satisfies $\Lambda(q) \neq 0$, for
all $q\in Q$. Then, we can consider the modular vector field
${\mathcal M}^{\nu \wedge \Lambda} \in {\frak X}(A^*)$ of the
linear Poisson structure on $A^*$ with respect to the volume form
$\nu \wedge \Lambda$ on $A^*$. ${\mathcal M}^{\nu \wedge \Lambda}$
is given by
\begin{equation}\label{MnuLam}
{\mathcal M}^{\nu \wedge \Lambda}(H) = div_{(\nu \wedge
\Lambda)}({\mathcal H}_{H}^{\Pi_{A^*}}), \mbox{ for } H \in
C^{\infty}(A^*),
\end{equation}
where $div_{(\nu \wedge \Lambda)}({\mathcal H}_{H}^{\Pi_{A^*}})$
is the divergence of the hamiltonian vector field ${\mathcal
H}_{H}^{\Pi_{A^*}}$ of $H$ with respect to the volume form $\nu
\wedge \Lambda$ (see \cite{We}).

We have that ${\mathcal M}^{(\nu \wedge \Lambda)}$ is the vertical
lift of a section ${\mathcal M}^{(\nu, \Lambda)} \in
\Gamma(\tau_{A^*})$, that is,
\begin{equation}\label{CamMod-SecMod}
{\mathcal M}^{(\nu \wedge \Lambda)} = ({\mathcal M}^{(\nu,
\Lambda)})^{\bf v}.
\end{equation}
${\mathcal M}^{(\nu, \Lambda)}$ is {\em the modular section} of
$A$ with respect to $\nu$ and $\Lambda$ (see \cite{EvLuWe,We}).

Denote by $\Omega_{\Lambda}$ the section of the vector bundle
$\tau_{\Lambda^nA^*}: \Lambda^nA^* \to Q$ characterized by
\[
\Omega_{\Lambda}(X_{1}, \dots , X_{n})\Lambda = X_{1} \wedge \dots
\wedge X_{n}, \mbox{ for } X_{1}, \dots , X_{n} \in
\Gamma(\tau_{A}).
\]
It follows that $\Omega_{\Lambda}(q) \neq 0$, for all $q \in Q$.
Moreover,
\begin{equation}\label{DefsecMo}
{\mathcal M}^{(\nu, \Lambda)}(X) = div_{\nu}(\rho_{A}(X)) -
div_{\Omega_{\Lambda}}X,
\end{equation}
where $div_{\Omega_{\Lambda}}X$ is the divergence of $X$ with
respect to $\Omega_{\Lambda}$, that is,
\[
{\mathcal L}_{X}^{A}\Omega_{\Lambda} =
(div_{\Omega_{\Lambda}}X)\Omega_{\Lambda}
\]
(for more details, see \cite{EvLuWe,We}).

If $(q^i)$ are local coordinates on an open subset $U$ of $Q$ and
$\{e_{\alpha}\}$ is a basis of sections of the vector bundle
$\tau_{A}^{-1}(U) \to U$ such that
\[
\nu = e^{\sigma_{U}^{\nu}}dq^1 \wedge \dots \wedge dq^m,
\makebox[.75cm]{} \Lambda = e_{1}\wedge \dots \wedge e_{n}
\]
then, from (\ref{DefsecMo}), we deduce that
\begin{equation}\label{Secmolo}
{\mathcal M}^{(\nu, \Lambda)} = (C^{\beta}_{\alpha\beta} +
\displaystyle \frac{\partial \rho^i_{\alpha}}{\partial q^i} +
\rho^i_{\alpha}\frac{\partial \sigma_{U}^{\nu}}{\partial
q^i})e^{\alpha}.
\end{equation}
We also have the following result
\begin{proposition}\cite{EvLuWe}
\begin{enumerate}
\item
${\mathcal M}^{(\nu, \Lambda)}$ is a $1$-cocycle, that is,
\[
d^A{\mathcal M}^{(\nu, \Lambda)} = 0.
\]
\item
The cohomology class of ${\mathcal M}^{(\nu, \Lambda)}$ doesn't
depend on the chosen volumes $\nu$ and $\Lambda$. In fact,
\begin{equation}\label{Cohomology}
{\mathcal M}^{(e^{\sigma}\nu, e^{\mu}\Lambda)} = {\mathcal
M}^{(\nu, \Lambda)} + d^A(\sigma + \mu), \mbox{ for } \sigma, \mu
\in C^{\infty}(Q).
\end{equation}
\end{enumerate}
\end{proposition}
The cohomology class of ${\mathcal M}^{(\nu, \Lambda)}$ is called
{\em the modular class} of $A$.

The Lie algebroid $\tau_{A}: A \to Q$ is said to be {\em
unimodular} if its modular class is zero, that is, there exists a
real $C^{\infty}$-function on $Q$, $\sigma: Q \to \R$, such that
\[
{\mathcal M}^{(\nu, \Lambda)} = -d^A\sigma.
\]
If $(q^i)$ are local coordinates on an open subset $U$ of $Q$ and
$\{e_{\alpha}\}$ is a basis of sections of the vector bundle
$\tau_{A}^{-1}(U) \to U$ such that
\[
\nu = e^{\sigma_{U}^{\nu}}dq^1\wedge \dots \wedge dq^m,
\makebox[.5cm]{} \Lambda = e_{1} \wedge \dots \wedge e_{n}
\]
then, using (\ref{Secmolo}), we deduce that ${\mathcal M}^{(\nu,
\Lambda)} = -d^A\sigma$ if and only if
\[
C_{\alpha\beta}^{\beta} + \displaystyle \frac{\partial
\rho^{i}_{\alpha}}{\partial q^i} + \rho^{i}_{\alpha}
\frac{\partial (\sigma_{U}^{\nu} + \sigma)}{\partial q^i} = 0,
\mbox{ for all } \alpha.
\]
Next, we will discuss the unimodularity of some example of Lie
algebroids.

Let $A$ be the standard Lie algebroid $\tau_{TQ}: TQ \to Q$. Then,
the linear Poisson structure on $\Pi_{T^*Q}$ is induced by the
canonical symplectic structure $\Omega_{T^*Q}$. In addition, it is
well-known that the hamiltonian vector fields on $T^*Q$ preserve
the symplectic volume. Thus, the Lie algebroid $\tau_{TQ}: TQ \to
Q$ is unimodular.

On the other hand, if ${\frak g}$ is a real Lie algebra of finite
dimension then, from (\ref{Secmolo}), it follows that the modular
section of ${\frak g}$ with respect to
\[
\Lambda = e_{1} \wedge \dots \wedge e_{n},
\]
$\{e_{\alpha}\}$ being a basis of ${\frak g}$, is just {\em the
modular character} of ${\frak g}$ (see \cite{EvLuWe}). Therefore,
our notion of a unimodular Lie algebra coincides with the
classical definition of a unimodular Lie algebra.

Now, we consider an Atiyah algebroid $\tau_{{\frak g} \times TM}:
{\frak g} \times TM \to M$, with ${\frak g}$ a real Lie algebra of
finite dimension and $M$ a smooth manifold (see Section
\ref{Sec2.2}). Let $\nu$ be a volume form on $M$ and $\{e_{a}\}$
be a basis of ${\frak g}$. Denote by $\chi_{\nu}$ the $m$-vector
on $M$ given by
\[
\chi_{\nu}(\alpha_{1}, \dots , \alpha_{m})\nu = \alpha_{1} \wedge
\dots \wedge \alpha_{m}, \mbox{ for } \alpha_{1}, \dots ,
\alpha_{m} \in \Omega^{1}(M).
\]
Then,
\[
\Lambda = \chi_{\nu} \wedge e_{1} \wedge \dots \wedge e_{n} \in
\Gamma(\Lambda^n({\frak g} \times TM)) \mbox{ and } \Lambda(x)
\neq 0, \mbox{ for all } x\in M.
\]
Moreover, from (\ref{DefsecMo}), it follows that
\[
{\mathcal M}^{(\nu, \Lambda)}(X) = 0, \; \; \; {\mathcal M}^{(\nu,
\Lambda)}(e_{a}) = {\mathcal M}_{\frak g}(e_{a}),
\]
for $X \in {\frak X}(M)$, where ${\mathcal M}_{\frak g}$ is the
modular character of ${\frak g}$. Note that, in this case,
\[
\Omega_{\Lambda} = \nu \wedge e^1 \wedge \dots \wedge e^n,
\]
$\{e^a\}$ being the dual basis to $\{e_{a}\}$.

Thus, the Lie algebroid $\tau_{{\frak g} \times TM}: {\frak g}
\times TM \to M$ is unimodular if and only if the Lie algebra
${\frak g}$ is unimodular.

Finally, suppose that $\Phi: {\frak g} \to {\frak X}(Q)$ is a left
infinitesimal action of ${\frak g}$ on the manifold $Q$ and that
$\tau_{{\frak g} \times Q}: {\frak g} \times Q \to Q$ is the
corresponding action Lie algebroid. Then, $\tau_{{\frak g} \times
Q}: {\frak g} \times Q \to Q$ is unimodular if and only if there
exists $\sigma \in C^{\infty}(Q)$ such that
\[
div_{\nu}\Phi(\xi) - \Phi(\xi)(\sigma) = {\mathcal M}_{\frak
g}(\xi), \mbox{ for } \xi \in {\frak g},
\]
where $\nu$ is a volume form on $Q$ and ${\mathcal M}_{\frak g}
\in {\frak g}^*$ is the modular character of ${\frak g}$. In
particular, if ${\frak g}$ is a unimodular Lie algebra and the
infinitesimal action preserves a volume form on $Q$ then
$\tau_{{\frak g} \times Q}: {\frak g} \times Q \to Q$ is a
unimodular Lie algebroid.

\section{Unimodularity and preservation of volumes}
Let $(\lcf \cdot, \cdot \rcf_{A}, \rho_{A})$ be a Lie algebroid
structure on a vector bundle $\tau_{A}: A \to Q$ of rank $n$ with
base manifold $Q$ of dimension $m$.

In this section we will assume that $Q$ is orientable and that the
vector bundle $\tau_{A}: A \to Q$ also is orientable.

Now, suppose that ${\mathcal G}$ is a bundle metric on $A$ and
that $V: Q \to \R$ is a real $C^{\infty}$-function on $Q$. Then,
we can consider {\em the hamiltonian energy} $H: A^* \to \R$ given
by
\begin{equation}\label{Hamener}
H(\alpha) = \frac{1}{2} {\mathcal G}(\alpha, \alpha) +
V(\tau_{A^*}(\alpha)), \mbox{ for } \alpha \in A^*.
\end{equation}
In other words, $H$ is the sum of the kinetic energy and the
potential energy.

We will denote by $\Lambda^{\mathcal G} \in
\Gamma(\tau_{\Lambda^nA})$ the volume induced by ${\mathcal G}$,
that is, $\Lambda^{\mathcal G}$ is characterized by the following
condition
\[
\Lambda^{\mathcal G}(q)(e_{q}^1, \dots , e_{q}^n) = 1, \mbox{ for
all } q \in Q,
\]
where $\{e_{q}^1, \dots , e_{q}^1\}$ is an orthonormal basis of
$A^*_q$ with positive orientation.

We will fix a volume form $\nu$ on $M$. Then, if $\Phi$ is a
volume form on $A^*$ we may suppose, without the loss of
generality, that
\[
\Phi = e^{\tilde{\sigma}}\nu \wedge \Lambda^{\mathcal G}, \mbox{
with } \tilde{\sigma} \in C^{\infty}(A^*).
\]
The volume $\Phi$ is said to be of {\em basic type} if
$\tilde{\sigma}$ is a basic function, that is, there exists a real
$C^{\infty}$-function $\mu: Q \to \R$ such that
\[
\mu \circ \tau_{A^*} = \tilde{\sigma}.
\]
Note that this definition doesn't depend on the chosen volume form
$\nu$. In fact, one may prove that $\Phi$ is a volume form of
basic type if and only if
\[
{\mathcal L}_{\alpha^{\bf v}}\Phi = 0, \mbox{ for } \alpha \in
\Gamma(\tau_{A^*}),
\]
${\mathcal L}$ being the Lie derivative operator.

Using the function $\tilde{\sigma}$ one may define the following
objects:
\begin{itemize}
\item
A real $C^{\infty}$-function $\sigma: Q \to \R$ given by
\[
\sigma(q) = \tilde{\sigma}(0(q)), \mbox{ for all } q\in Q,
\]
where $0: Q \to A^*$ is the zero section of $\tau_{A^*}: A^* \to
Q$.
\item
The vertical derivative of $\tilde{\sigma}$
\[
\mathbb{F}\tilde{\sigma}: A^* \to A
\]
defined by
\begin{equation}\label{Derver}
\beta_q(\mathbb{F}\tilde{\sigma}(\alpha_q)) = \displaystyle
\frac{d}{dt}_{|t=0} \tilde{\sigma}(\alpha_q + t\beta_q), \mbox{
for } \alpha_q, \beta_q \in A^*_q.
\end{equation}

\end{itemize}
Now, we will prove the main result of this paper.
\begin{theorem}\label{maintheorem}
Let $H: A^* \to \R$ be the hamiltonian energy induced by a bundle
metric ${\mathcal G}$ on $A$ and a potential energy $V: Q \to \R$.
\begin{enumerate}
\item
If $A$ is unimodular then the hamiltonian vector field ${\mathcal
H}_{H}^{\Pi_{A^*}}$ of $H$ with respect to the linear Poisson
structure $\Pi_{A^*}$ preserves a volume form on $A^*$ of basic
type. In fact, if ${\mathcal M}^{(\nu, \Lambda^{\mathcal G})} =
-d^{A}\sigma$, with $\sigma \in C^{\infty}(Q)$ and ${\mathcal
M}^{(\nu, \Lambda^{\mathcal G})}$ the modular section of $A$ with
respect to the volumes $\nu$ and $\Lambda^{\mathcal G}$, we have
that ${\mathcal H}_{H}^{\Pi_{A^*}}$  preserves the volume form on
$A^*$
\[
\Phi = e^{\sigma}\nu \wedge \Lambda^{\mathcal G}.
\]

\item
If the hamiltonian vector field ${\mathcal H}_{H}^{\Pi_{A^*}}$
preserves a volume form $\Phi$ on $A^*$,
\[
\Phi = e^{\tilde{\sigma}}\nu \wedge \Lambda^{\mathcal G},
\]
then
\[
(\flat_{\mathcal G}^{-1}({\mathcal M}^{(e^{\sigma}\nu,
\Lambda^{\mathcal G})}))^{\bf v} \circ 0 =
(T\mathbb{F}\tilde{\sigma} \circ (d^AV)^{\bf v})\circ 0
\]
where ${\mathcal M}^{(e^{\sigma}\nu, \Lambda^{\mathcal G})}$ is
the modular section of $A$ with respect to the volumes
$e^{\sigma}\nu$ and $\Lambda^{\mathcal G}$, ${\bf v}$ is the
vertical lift, $T\mathbb{F}\tilde{\sigma}: TA^* \to TA$ is the
tangent map to $\mathbb{F}\tilde{\sigma}: A^* \to A$ and
$\flat_{\mathcal G}: \Gamma(\tau_{A}) \to \Gamma(\tau_{A^*})$ is
the isomorphism of $C^{\infty}(Q)$-modules induced by ${\mathcal
G}$.
\end{enumerate}
\end{theorem}
\begin{proof}
(i) Since $A$ is unimodular, we deduce that there exists a real
$C^{\infty}$-function $\mu$ on $Q$ such that
\[
{\mathcal M}^{(\nu, \Lambda^{\mathcal G})} = -d^A\mu.
\]
Thus, from (\ref{Cohomology}), it follows that
\[
{\mathcal M}^{(e^{\mu}\nu, \Lambda^{\mathcal G})} = 0.
\]
This implies that ${\mathcal M}^{(e^{\mu}\nu \wedge
\Lambda^{\mathcal G})} = ({\mathcal M}^{(e^{\mu}\nu,
\Lambda^{\mathcal G})})^{\bf v} = 0$ (see (\ref{CamMod-SecMod})).
In particular,
\[
0 = {\mathcal M}^{(e^{\mu}\nu \wedge\Lambda^{\mathcal G})}(H) =
div_{(e^{\mu}\nu\wedge \Lambda^{\mathcal G})}({\mathcal
H}_{H}^{\Pi_{A^*}}).
\]
Therefore, the hamiltonian vector field ${\mathcal
H}_{H}^{\Pi_{A^*}}$ preserves the volume form of basic type
\[
\Phi = e^{\mu}\nu \wedge \Lambda^{\mathcal G}.
\]

(ii) Suppose that $(q^i)$ are local coordinates on $Q$ such that
\[
\nu = e^{\sigma_{U}^{\nu}} dq^1 \wedge \dots \wedge dq^m
\]
and that $\{e_{1}, \dots , e_{n}\}$ is a local orthonormal basis
of $\Gamma(A)$ with positive orientation. Then,
\[
\Lambda^{\mathcal G} = e_{1} \wedge \dots \wedge e_{n}.
\]
Moreover, if $(q^i, p_{\alpha})$ are the corresponding local
coordinates on $A^*$, we have that (see (\ref{Hamener}) and the
proof of Lemma \ref{Voldual})
\begin{equation}\label{LocHPhi}
H(q^i, p_{\alpha}) = \displaystyle \frac{1}{2} \sum_{\alpha}
(p_{\alpha})^2 + V(q^i), \makebox[.5cm]{} \Phi =
e^{\tilde{\sigma}}e^{\sigma_{U}^{\nu}} dq^1 \wedge \dots \wedge
dq^m \wedge dp_1 \wedge \dots \wedge dp_n.
\end{equation}
Thus, from (\ref{MnuLam}), (\ref{CamMod-SecMod}), (\ref{Secmolo})
and (\ref{LocHPhi}), it follows that
\[
\begin{array}{rcl}
0&=& e^{\tilde{\sigma}}e^{\sigma_{U}^{\nu}}{\mathcal
H}_{H}^{\Pi_{A^*}}(\tilde{\sigma})
dq^1 \wedge \dots \wedge dq^m \wedge dp_1 \wedge \dots \wedge dp_n \\[5pt]
&& + \displaystyle
e^{\tilde{\sigma}}e^{\sigma_{U}^{\nu}}(\frac{\partial
\rho^i_{\alpha}}{\partial q^i} + C_{\alpha\beta}^{\beta}+
\rho^i_{\alpha} \frac{\partial \sigma_{U}^{\nu}}{\partial
q^i})p_{\alpha} dq^1 \wedge \dots \wedge dq^m \wedge dp_1 \wedge
\dots \wedge dp_n.
\end{array}
\]
Therefore, using (\ref{HamH}) and (\ref{LocHPhi}), we deduce that
\[
\displaystyle \rho^i_{\alpha}p_{\alpha}\frac{\partial
\tilde{\sigma}}{\partial q^i} - (\frac{\partial V}{\partial
q^i}\rho^i_{\alpha}\frac{\partial \tilde{\sigma}}{\partial
p_{\alpha}} +
C_{\alpha\beta}^{\gamma}p_{\beta}p_{\gamma}\frac{\partial
\tilde{\sigma}}{\partial p_{\alpha}}) + (\frac{\partial
\rho^i_{\alpha}}{\partial q^i} + C_{\alpha\beta}^{\beta} +
\rho^i_{\alpha} \frac{\partial \sigma_{U}^{\nu}}{\partial
q^i})p_{\alpha} = 0.
\]
Now, if we take the derivative of the above expression with
respect to the variable $p_{\mu}$, we obtain that
\[
\begin{array}{rcl}
0&=& \displaystyle \rho^i_{\mu}\frac{\partial \tilde{\sigma}}{\partial q^i}
+ \rho^i_{\alpha} p_{\alpha} \frac{\partial^2\tilde{\sigma}}{\partial q^i \partial p_{\mu}}
-\frac{\partial V}{\partial q^i}\rho^i_{\alpha}
\frac{\partial^2\tilde{\sigma}}{\partial p_{\alpha} \partial p_{\mu}}
- C_{\alpha\mu}^{\gamma} p_{\gamma} \frac{\partial \tilde{\sigma}}{\partial p_{\alpha}} \\[5pt]
&& \displaystyle -C_{\alpha\beta}^{\mu}p_{\beta}\frac{\partial
\tilde{\sigma}}{\partial p_{\alpha}}-
C_{\alpha\beta}^{\gamma}p_{\beta}p_{\gamma}\frac{\partial^2\tilde{\sigma}}{\partial
p_{\alpha}\partial p_{\mu}} + \frac{\partial \rho^i_\mu}{\partial
q^i} + C_{\mu\beta}^{\beta} + \rho^i_{\mu} \frac{\partial
\sigma_{U}^{\nu}}{\partial q^i}.
\end{array}
\]
Consequently, along the zero section $0: M \to A$, we have that
\begin{equation}\label{Keypoint}
\displaystyle \rho^i_{\mu}\frac{\partial (\sigma_{U}^{\nu}
+\sigma)}{\partial q^i} + \frac{\partial \rho^i_\mu}{\partial q^i}
+ C_{\mu\beta}^{\beta} = \frac{\partial V}{\partial
q^i}\rho^i_{\alpha} (\frac{\partial^2\tilde{\sigma}}{\partial
p_{\alpha} \partial p_{\mu}})_{|0(Q)}, \mbox{ for all } \mu.
\end{equation}
On the other hand, from (\ref{Derver}), it follows that
\[
\mathbb{F}\tilde{\sigma}(q^i, p_{\alpha}) = (q^i, \displaystyle
\frac{\partial \tilde{\sigma}}{\partial p_{\alpha}}).
\]
This implies that
\begin{equation}\label{Aux1}
\displaystyle T\mathbb{F}\tilde{\sigma} \circ (d^AV)^{\bf v} =
(\frac{\partial V}{\partial q^i}\rho^i_{\alpha}
\frac{\partial^2\tilde{\sigma}}{\partial p_{\alpha} \partial
p_{\mu}})\frac{\partial}{\partial p_{\mu}}.
\end{equation}
In addition, since $\{e_{\alpha}\}$ is an orthonormal basis, we
have that
\[
\flat_{\mathcal G}(e_{\alpha}) = e^{\alpha}, \; \; \mbox{ for all
} \alpha.
\]
Thus, from (\ref{Secmolo}) and (\ref{Cohomology}), we deduce that
\begin{equation}\label{Aux2}
(\flat_{\mathcal G}^{-1}({\mathcal M}^{(e^{\sigma}\nu,
\Lambda^{\mathcal G})}))^{\bf v} = (\displaystyle
\rho^i_{\mu}\frac{\partial (\sigma + \sigma_{U}^{\nu})}{\partial
q^i} + \frac{\partial \rho^i_\mu}{\partial q^i} +
C_{\mu\beta}^{\beta})\frac{\partial}{\partial p_{\mu}}.
\end{equation}
Therefore, using (\ref{Keypoint}), (\ref{Aux1}) and (\ref{Aux2}),
we prove the result.
\end{proof}
Note that if $\tilde{\sigma}$ is a basic function then, from
(\ref{Aux1}), it follows that
\[
\displaystyle T\mathbb{F}\tilde{\sigma} \circ (d^AV)^{\bf v} = 0.
\]
Consequently, using Theorem \ref{maintheorem}, we obtain that
\begin{corollary}\label{Corolario1}
Let $H: A^* \to \R$ be the hamiltonian energy induced by a bundle
metric on $A$ and a potential energy $V: Q \to \R$. Then, $A$ is
unimodular if and only if the hamiltonian vector field ${\mathcal
H}_{H}^{\Pi_{A^*}}$ of $H$ preserves a volume form on $A^*$ of
basic type. In fact, if ${\mathcal M}^{(\nu, \Lambda^{\mathcal
G})} = -d^{A}\sigma$, with $\sigma \in C^{\infty}(Q)$ and
${\mathcal M}^{(\nu, \Lambda^{\mathcal G})}$ the modular section
of $A$ with respect to the volumes $\nu$ and $\Lambda^{\mathcal
G}$, we have that ${\mathcal H}_{H}^{\Pi_{A^*}}$ preserves the
volume form on $A^*$
\[
\Phi = e^{\sigma}\nu \wedge \Lambda^{\mathcal G}.
\]
\end{corollary}
From Theorem \ref{maintheorem} we also obtain the following result
\begin{corollary}\label{Corolario2}
Let $H: A^* \to \R$ be the kinetic energy induced by a bundle
metric on $A$ (in this case the potential energy is constant).
Then, $A$ is unimodular if and only if the hamiltonian vector
field ${\mathcal H}_{H}^{\Pi_{A^*}}$ preserves a volume form on
$A^*$. In fact, if ${\mathcal M}^{(\nu, \Lambda^{\mathcal G})} =
-d^{A}\sigma$, with $\sigma \in C^{\infty}(Q)$ and ${\mathcal
M}^{(\nu, \Lambda^{\mathcal G})}$ the modular section of $A$ with
respect to the volumes $\nu$ and $\Lambda^{\mathcal G}$, we have
that ${\mathcal H}_{H}^{\Pi_{A^*}}$ preserves the volume form on
$A^*$
\[
\Phi = e^{\sigma}\nu \wedge \Lambda^{\mathcal G}.
\]
\end{corollary}

\section{Examples}

\subsection{Standard mechanical hamiltonian systems}
Let $Q$ be a smooth manifold of dimension $m$ and $H:T^*Q \to \R$
be a hamiltonian function,
\[
H(\alpha) = \displaystyle \frac{1}{2} {\mathcal G}(\alpha, \alpha)
+ V(\tau_{T^*Q}(\alpha)), \mbox{ for } \alpha \in T^*Q,
\]
with ${\mathcal G}$ a Riemannian metric on $Q$ and $V:Q \to \R$
the potential energy. In this case, our Lie algebroid is the
standard one $\tau_{TQ}: TQ \to Q$. As we know (see Subsection
\ref{modularclass}), $\tau_{TQ}: TQ \to Q$ is a unimodular Lie
algebroid and, thus, the hamiltonian vector field ${\mathcal
H}_{H}^{\Pi_{T^*Q}}$ preserves a volume form $\Phi$ on $T^*Q$ of
basic type (see Theorem \ref{maintheorem}). In fact, we may take
$\Phi = \Omega_{Q}^{m}$, where $\Omega_{Q}$ is the canonical
symplectic structure of $T^*Q$ (note that if $Z$ is an arbitrary
hamiltonian vector field on $T^*Q$ then ${\mathcal
L}_{Z}\Omega_{Q} = 0$).

\subsection{Mechanical hamiltonian systems on the dual bundle to
an Atiyah algebroid} Let $Q = G \times M$ be total space of a
tricial principal $G$-bundle over a manifold $M$ of dimension $m$.
Then, the Atiyah algebroid associated with the principal bundle is
the vector bundle $\tau_{{\frak g} \times TM}: {\frak g} \times TM
\to M$, where ${\frak g}$ is the Lie algebra of $G$. Thus, if
${\mathcal G}$ is a bundle metric on $\tau_{{\frak g} \times TM}:
{\frak g} \times TM \to M$ and $V: M \to \R$ is a real
$C^{\infty}$-function on $M$, we may consider the hamiltonian
function $H: {\frak g}^* \times T^*M \to \R$ given by
\[
H(\alpha, \beta) = \displaystyle \frac{1}{2} {\mathcal G}((\alpha,
\beta), (\alpha, \beta)) + V(\tau_{T^*M}(\beta)),
\]
for $(\alpha, \beta) \in {\frak g}^* \times T^*M$.

Denote by ${\mathcal H}_{H}^{\Pi_{{\frak g}^* \times T^*M}}$ the
hamiltonian vector field of $H$ with respect to the linear Poisson
structure $\Pi_{{\frak g}^* \times T^*M}$ on ${\frak g}^* \times
T^*M$. Then, using Corollary \ref{Corolario1} (see also Subsection
\ref{modularclass}), we deduce that the vector field ${\mathcal
H}_{H}^{\Pi_{{\frak g}^* \times T^*M}}$ preserves a volume form
$\Phi$ on ${\frak g}^* \times T^*M$ of basic type if and only if
${\frak g}$ is a unimodular Lie algebra. In fact, if ${\frak g}$
is unimodular and $\{e_{\alpha}\}$ is a basis of ${\frak g}$ then
we may take
\[
\Phi = dp_{1} \wedge \dots \wedge dp_{n} \wedge \Omega_{M}^{m}
\]
where $p_{\alpha}$ are the (global) coordinates on ${\frak g}^*$
induced by the basis $\{e_{\alpha}\}$ and $\Omega_{M}$ is the
canonical symplectic structure on $T^*M$.

Note that in the particular case when $M$ is a single point then
the potential energy $V$ is a constant and the linear Poisson
structure on ${\frak g}^* \times T^*M \simeq {\frak g}^*$ is just
the Lie-Poisson structure of ${\frak g}^*$. Thus, we recover a
result which was proved by Kozlov \cite{Ko}: {\em the hamiltonian
vector field ${\mathcal H}_{H}^{\Pi_{{\frak g}^*}}$ preserves a
volume form on ${\frak g}^*$ if and only if ${\frak g}$ is
unimodular}.

A very simple example of a mechanical system on an Atiyah
algebroid is {\em the Elroy's beanie}: two planar rigid bodies
attached at their centers of mass, moving freely in the plane (see
\cite{Le,Os}). In this case:
\begin{itemize}
\item
The manifold $M$ is $S^1$ and the Lie group $G = SE(2)$.
\item
The bundle metric ${\mathcal G}$ on $\tau_{{\frak se}(2) \times
TS^1}: {\frak se}(2) \times TS^1 \simeq \R^3 \times (S^1 \times
\R) \to S^1$ is given by
\[
{\mathcal G}(\theta)(((\xi_{1}, \xi_{2}, \xi_{3}), t), ((\xi'_{1},
\xi'_{2}, \xi'_{3}), t')) = m(\xi_{1}\xi'_{1} + \xi_{2}\xi'_{2}) +
(I_{1} + I_{2})\xi_{3}\xi'_{3} + I_{2}tt' + I_{2}(\xi_{3}t' +
t\xi'_{3})
\]
for $\theta \in S^1$, $(\xi_{1},\xi_{2},\xi_{3}), (\xi'_{1},
\xi'_{2}, \xi'_{3}) \in {\frak se}(2) \simeq \R^3$ and $t, t' \in
\R$, where $m$, $I_1$ and $I_{2}$ are constants.

\end{itemize}

Since ${\frak se}(2)$ is a unimodular Lie algebra, we deduce that
the hamiltonian dynamics on ${\frak se}(2)^{*} \times T^*S^1
\simeq \R^3 \times (S^1 \times \R)$ preserves the volume form
\[
\Phi = dp_{1} \wedge dp_{2} \wedge dp_{3} \wedge \psi \wedge dt,
\]
where $\psi$ is the length element of $S^1$.

\subsection{Mechanical hamiltonian systems on the dual bundle to
an action Lie algebroid} Let $\Phi: {\frak g} \to {\frak X}(Q)$ be
a left infinitesimal action of ${\frak g}$ on a manifold $Q$ and
$\tau_{{\frak g}\times Q}: {\frak g}\times Q \to Q$ be the
corresponding action Lie algebroid. Suppose that $<\cdot, \cdot>$
is a scalar product on ${\frak g}$, that $V: Q \to \R$ is a real
$C^{\infty}$-function on $Q$ and that $H: {\frak g}^* \times Q\to
\R$ is the corresponding hamiltonian function given by
\[
H(\alpha, q) = \displaystyle \frac{1}{2} <\alpha, \alpha> + V(q),
\mbox{ for } \alpha \in {\frak g}^* \mbox{ and } q\in Q.
\]
Denote by ${\mathcal H}_{H}^{\Pi_{{\frak g}^* \times Q}}$ the
hamiltonian vector field of $H$ with respect to the linear Poisson
structure $\Pi_{{\frak g}^*\times Q}$ on ${\frak g}^*\times Q$.
Then, using Corollary \ref{Corolario2} (see also Subsection
\ref{modularclass}) we deduce that the hamiltonian vector field
${\mathcal H}_{H}^{\Pi_{{\frak g}^* \times Q}}$ preserves a volume
form $\Phi$ on ${\frak g}^* \times Q$ of basic type if and only if
there exists a real $C^{\infty}$-function $\sigma: Q \to \R$ such
that
\begin{equation}\label{vol-action}
div_{\nu}\Phi(\xi) - \Phi(\xi)(\sigma) = {\mathcal M}_{\frak
g}(\xi), \mbox{ for all } \xi \in {\frak g},
\end{equation}
where $\nu$ is an arbitrary volume form on $Q$ and ${\mathcal
M}_{\frak g}$ is the modular character of ${\frak g}$. Moreover,
if (\ref{vol-action}) holds and $\{e_{\alpha}\}$ is a basis of
${\frak g}$ we may take
\[
\Phi = e^{-\sigma}dp_{1} \wedge \dots \wedge dp_{n} \wedge \nu
\]
where $(p_{\alpha})$ are the global coordinates on ${\frak g}^*$
induced by the basis $\{e_{\alpha}\}$.

Note that if the Lie algebra ${\frak g}$ is unimodular (that is,
${\mathcal M}_{\frak g} = 0$) and the left infinitesimal action
preserves the volume form $\nu$ (that is, $div_{\nu}\Phi(\xi) =
0$, for all $\xi \in {\frak g}$) then (\ref{vol-action}) holds (we
may take $\sigma = 0$). As a particular example we can consider an
interesting mechanical system: {\em the heavy-top} (see
\cite{Ma}). In this case:
\begin{itemize}
\item
${\frak g} = {\frak so}(3) \simeq \R^3$ is the Lie algebra of the
special orthogonal group $SO(3)$ and $M$ is the sphere $S^2$.
\item
$\Phi: {\frak so}(3) \to S^2$ is the left infinitesimal action
induced by the standard left action of the Lie group $SO(3)$ on
$S^2$.
\item
The scalar product on ${\frak g} = {\frak so}(3)$ is given by
\[
<\omega, \omega'> = \omega \cdot I\omega,
\]
$I$ being the inertia tensor of the top.
\item
The potential $V: S^2 \to \R$ is defined by
\[
V(x) = mgl x\cdot e, \mbox{ for } x\in S^2,
\]
where $e$ is the unit vector from the fixed point to the center of
mass and $m$, $g$ and $l$ are constants.

\end{itemize}

Note that the action of $SO(3)$ on $S^2$ preserves the symplectic
volume $\nu$ of $S^2$ and that ${\frak g} = {\frak so}(3)$ is a
unimodular Lie algebra. Thus, the hamiltonian dynamics preserves
the volume $\Phi$ on ${\frak g}^* \times S^2 \simeq {\frak
so}(3)^* \times S^2 \simeq \R^3 \times S^2$ given by
\[
\Phi = dp_{1} \wedge dp_{2} \wedge dp_{3} \wedge \nu.
\]

\section{Conclusions and outlook}
We have proved that a hamiltonian system of mechanical type on the
dual bundle to a Lie algebroid $A$ preserves a volume form on
$A^*$ of basic type if and only if $A$ is unimodular. In addition,
if the potential energy of the hamiltonian function is constant
then we deduce that the hamiltonian dynamics preserves a volume
form on $A^*$ (not necessarily of basic type) if and only if $A$
is unimodular. These results generalize Liouville's theorem (see
\cite{AbMa}) and a previous result for Lie-Poisson equations which
was proved by Kozlov \cite{Ko}.

On the other hand, after we finished this paper, we obtained some
examples of hamiltonian systems of mechanical type (with
non-constant potential energy) on the dual bundle to a
non-unimodular Lie algebroid which preserve a volume form of
non-basic type. So, it would be interesting to discuss the
existence of such volumes on non-unimodular Lie algebroids. This
will be the subject of a separate publication.

Another goal we have proposed is to extend the results of this
paper for non-holonomic mechanical systems on Lie algebroids (see
\cite{FeMa}). Previous results in this direction for some
particular classes of Lie algebroids have obtained by several
authors (see \cite{Jo,Ko,Ko1,ZeBl}).

\end{document}